\RequirePackage{amsmath}
\documentclass[runningheads,envcountsect,envcountsame]{llncs}
\usepackage{todonotes}
\usepackage[T1]{fontenc}
\usepackage[table]{xcolor}
\usepackage[noadjust]{cite}
\usepackage[ruled,linesnumbered,noend]{algorithm2e}
\usepackage{orcidlink,amssymb,enumitem,calc,mdframed,multirow,booktabs}

\usepackage{tikz}
\usetikzlibrary{automata,arrows.meta,shapes.geometric,calc,fit}
\tikzset{->,>=Stealth,every state/.style={thick}}

\newcommand{\game}{{\mathcal{G}}}
\newcommand{\sq}{{\square}}
\newcommand{\tri}{{\triangle}}
\newcommand{\EvenState}{{S_\sq}}
\newcommand{\OddState}{{S_\tri}}
\newcommand{\EvenStrategy}{{\sigma_\sq}}
\newcommand{\OddStrategy}{{\sigma_\tri}}
\newcommand{\strategies}{{\mathrm{\Sigma}}}
\newcommand{\lfunc}{{\mathrm{\Psi}}}
\newcommand{\bellman}{{\mathrm{\Phi}}}
\newcommand{\init}{{s_\mathrm{init}}}
\newcommand{\bvi}{{\mathrm{\Theta}}}

\newcommand{\dist}{{\textup{\textsc{Dist}}}}
\newcommand{\av}{{\textup{\textsc{Av}}}}
\newcommand{\post}{{\textup{\textsc{Post}}}}
\newcommand{\EvenStrategies}{{\textup{\textsc{Str}}_\sq}}
\newcommand{\OddStrategies}{{\textup{\textsc{Str}}_\tri}}
\newcommand{\infi}{{\textup{\textsc{Inf}}}}
\newcommand{\parity}{{\textup{\textsc{Parity}}}}
\newcommand{\attr}{{\textup{\textsc{Attr}}}}
\newcommand{\exit}{{\textup{\textsc{Exit}}}}

\newcommand{\mypar}[1]{\noindent\textbf{#1}}
\newcommand{\myparskip}[1]{\medskip\noindent\textbf{#1}}
\spnewtheorem{mynotation}[theorem]{Notation}{\itshape}{\mdseries}
\spnewtheorem{mydefinition}[theorem]{Definition}{\bfseries}{\mdseries}
\spnewtheorem{myproblem}[theorem]{Problem}{\bfseries}{\mdseries}

\begin{document}

\title{Value Iteration for Stochastic Parity Games}
\author{Kittiphon Phalakarn\inst{1}\orcidlink{0009-0006-5406-7480} \and
Ichiro Hasuo\inst{1,2,3}\orcidlink{0000-0002-8300-4650}}
\authorrunning{K. Phalakarn and I. Hasuo}
\institute{National Institute of Informatics, Tokyo, Japan\\\email{\{kphalakarn,hasuo\}@nii.ac.jp} \and
SOKENDAI (The Graduate University for Advanced Studies), Kanagawa, Japan \and
Imiron Co., Ltd., Tokyo, Japan}
\maketitle

\begin{abstract}
We present the first \emph{(bounded) value iteration} algorithm for the quantitative analysis of \emph{stochastic parity games}, a fundamental model for probabilistic verification with $\omega$-regular objectives. Existing algorithms are based on \emph{strategy iteration}, which repeatedly computes optimal strategies for one player while fixing the other, leading to high computational cost. Our value iteration algorithm instead operates directly on a lattice-theoretic characterization of winning probabilities, exploiting structural properties of (almost-sure qualitative) winning states under parity objectives. We prove correctness and convergence of the proposed algorithm. 
\keywords{Stochastic parity game \and Bounded value iteration \and Probabilistic model checking \and Fixed point.}
\end{abstract}

\section{Introduction}\label{sec:intro}

\mypar{Stochastic Parity Games.} Stochastic parity games (SPGs) are zero-sum, non-terminating games played on probabilistic transition systems between two players, Even (denoted by $\sq$) and Odd (denoted by $\tri$).\footnote{The symbol $\sq$ has an even number of sides, while $\tri$ has an odd number.} Each state is assigned a natural number, called a \emph{priority}. At each step, the player controlling the current state selects an available \emph{action}, after which the game moves to a successor state according to a predefined probability distribution decided by the chosen action.

Under the \emph{parity objective}, an infinite path is winning for Even if the \emph{minimum} priority appearing infinitely often along the path is even; otherwise, the path is winning for Odd. Accordingly, each player aims to select actions so as to optimize their probability of winning.

Stochastic parity games play a central role in several areas of computer science. Prior theoretical work establishes that SPGs are determined, meaning that for every state the winning probabilities of the two players sum to one \cite{DBLP:journals/jcss/AlfaroM04,DBLP:journals/jsyml/Martin98}. Moreover, SPGs admit optimal strategies that are pure (deterministic) and memoryless (positional) \cite{DBLP:conf/soda/ChatterjeeJH04,DBLP:conf/fossacs/Zielonka04}. Consequently, the decision problem for SPGs lies in the complexity class $\mathbf{NP} \cap \mathbf{coNP}$.

In addition, solving SPGs is polynomial-time equivalent to solving stochastic reachability, mean-payoff, discounted-payoff, and terminal-payoff games \cite{DBLP:conf/isaac/AnderssonM09,DBLP:conf/concur/BerthonKZ25,DBLP:journals/ipl/ChatterjeeH08}. Due to their expressiveness, SPGs have numerous applications, including reasoning about cyber-physical systems, economic models, computer networks, software systems, and security \cite{DBLP:journals/ejcon/SvorenovaK16,DBLP:journals/access/TusharYSNASP23,DBLP:journals/arcras/MardenS18}.

\myparskip{Previous Works: Strategy Iteration.} Despite extensive theoretical results on SPGs, only a few algorithms are available for quantitatively solving SPGs, i.e., computing the winning probability from a given initial state. The first such approach was proposed by Chatterjee and Henzinger \cite{DBLP:conf/stacs/ChatterjeeH06}, followed by the other by Hahn et al. \cite{DBLP:conf/vmcai/HahnST017}.

Both approaches are based on the same underlying technique, namely \emph{strategy iteration}. The algorithm starts with an arbitrary strategy for one player, say Even. By fixing Even's strategy, the SPG reduces to a Markov decision process (MDP). An optimal strategy for the other player, Odd, can then be computed for this MDP using existing techniques, e.g., \cite{DBLP:phd/us/Alfaro97,DBLP:journals/tac/CourcoubetisY98}.

The resulting strategy for Odd is subsequently used to improve Even's strategy so as to increase Even's winning probabilities. The improved Even's strategy is then used in the next iteration, where Odd's optimal strategy is recomputed and Even's strategy is further refined. This process repeats until Even's strategy cannot be improved, at which point it is optimal and the corresponding winning probabilities are obtained.

Due to its nature, strategy iteration can be computationally expensive, as it requires constructing optimal strategies in each iteration. Indeed, Hahn et al.~\cite{DBLP:conf/vmcai/HahnST017} reported that a significant portion of the overall runtime is spent on repeatedly solving MDPs.

\myparskip{Bounded Value Iteration.} Another widely used technique for computing winning probabilities is \emph{value iteration} \cite{DBLP:conf/spin/ChatterjeeH08}. This technique has been applied to solve parity games on Markov chains and MDPs, but, to the best of our knowledge, not SPGs. In practical settings where approximate probabilities suffice, value iteration is often preferred due to its favorable performance characteristics \cite{DBLP:conf/tacas/HartmannsJQW23,revised}.

Value iteration is based on a lattice-theoretic fixed-point characterization of winning probabilities. The algorithm initializes a value for each state and repeatedly updates these values using local update rules. Owing to this local nature of the value updates, value iteration is generally more efficient in practice than global strategy-based methods.

For stochastic reachability games (also known as simple stochastic games \cite{DBLP:journals/iandc/Condon92}), several variants of value iteration have been proposed \cite{DBLP:conf/cav/KelmendiKKW18,DBLP:journals/iandc/EisentrautKKW22,DBLP:conf/lics/KretinskyMW23,DBLP:conf/cav/MeggendorferW24,DBLP:conf/cav/PhalakarnTHH20,DBLP:conf/atva/PhalakarnTH25,DBLP:conf/atva/AzeemEKSW22,DBLP:conf/gandalf/AzeemKW25}. To provide precision guarantees, these algorithms compute both lower and upper approximations of the winning probabilities, a technique commonly referred to as \emph{bounded value iteration (BVI)}.

\myparskip{This Work: The First BVI for SPGs.} In this work, we present the first bounded value iteration algorithm for the quantitative computation of winning probabilities in stochastic parity games. The main challenge stems from the fact that these winning probabilities do not correspond to the least or greatest fixed points of the one-step Bellman operator (cf.~Definition~\ref{def:bellman}), in contrast to simpler objectives such as reachability or safety. Consequently, value updates based solely on local reasoning are insufficient for deriving converging lower and upper bounds.

To address this issue, we incorporate global (qualitative) information, namely the (almost-sure qualitative) winning regions of each player under the parity objective in suitably defined subgames. These subgames evolve dynamically with the current lower and upper approximations. Under an optimal play, Even avoids remaining in states that are almost-sure winning for Odd, and vice versa. Our value update operators are designed to capture this strategic behavior by adjusting values in accordance with these global considerations.


Although the structure of our algorithm may resemble that of~\cite[Algorithm~2]{DBLP:conf/lics/KretinskyMW23}, ours is based on fundamentally different principles and do not rely on any assumptions required in that work. Specifically, \cite{DBLP:conf/lics/KretinskyMW23} requires that the winning probabilities correspond to the least or greatest fixed point of the one-step Bellman operator, whereas ours makes no such assumption; moreover, their approach relies on reasoning over MDPs, while ours works directly on SPGs. More details are given in Section~\ref{subsec:lics23}.

\myparskip{Contributions.} Our main technical contribution is the first BVI algorithm for SPGs called \textsc{SPG-BVI} (Algorithm~\ref{alg:ec-spg-bvi}). We establish its correctness and convergence, and provide a detailed illustrative example of its execution.

\myparskip{Organization.} The remainder of the paper is structured as follows. Section~\ref{sec:prelim} introduces the necessary preliminaries, and Section~\ref{sec:related} surveys related work. In Section~\ref{sec:alg}, we present our \textsc{SPG-BVI} algorithm for SPGs together with proofs of correctness and convergence. 
Section~\ref{sec:conclusion} concludes the paper.

\section{Preliminaries}\label{sec:prelim}

\begin{mynotation}
For a finite set $S$, the set of all functions from $S$ to $[0,1]$ is denoted by $[0,1]^S$, and the set of all discrete probability distributions on $S$ is denoted by $\dist(S) := \{ d \in [0,1]^S : \sum_{s \in S} d(s) = 1 \}$.
\end{mynotation}

\subsection{Stochastic Parity Games}\label{subsec:game}

Our definition on stochastic parity games is as follows.\footnote{We do not introduce the player \emph{Random} as in \cite{DBLP:conf/csl/ChatterjeeJH03}. Nonetheless, our definition and that of \cite{DBLP:conf/csl/ChatterjeeJH03} are equivalent with respect to parity objectives.}

\begin{mydefinition}[stochastic parity game $\game$]\label{def:game}
A \emph{stochastic parity game (SPG)} is $\game = (S, \EvenState, \OddState, A, \av, \delta, \mathcal{P})$, where $S = \EvenState \uplus \OddState$ is a finite set of \emph{states} partitioned into \emph{Even's} $(\EvenState)$ and \emph{Odd's} $(\OddState)$ states, $A$ is a finite set of \emph{actions}, $\av : S \to 2^A \setminus \{\emptyset\}$ defines \emph{available} actions, $\delta : S \times A \to \dist(S)$ is a \emph{transition function}, and $\mathcal{P} : S \to \mathbb{N}$ is a \emph{priority function}.
\end{mydefinition}

The semantics of SPGs is defined as usual \cite{DBLP:conf/csl/ChatterjeeJH03,DBLP:books/daglib/0020348}. Let $\post(s,a) := \{ s' \in S : \delta(s,a,s') > 0 \}$. An infinite sequence $\rho = s_0 a_0 s_1 a_1 \ldots \in (S \times A)^\omega$ is an \emph{infinite path} through $\game$ if $a_i \in \av(s_i)$ and $s_{i+1} \in \post(s_i,a_i)$ for all $i \in \mathbb{N}$. A finite sequence $\rho = s_0 a_0 s_1 a_1 \ldots s_n \in (S \times A)^* \times S$ is a \emph{finite path} through $\game$ if it is a prefix of an infinite path.

A \emph{strategy $\EvenStrategy$ of Even} is a function $\EvenStrategy : \EvenState \to A$ such that $\EvenStrategy(s) \in \av(s)$ for all $s \in \EvenState$. A \emph{strategy $\OddStrategy$ of Odd} is defined analogously. We restrict our consideration to strategies of this form (i.e., pure memoryless strategies) as they are complete for finite SPGs with the parity objective \cite{DBLP:conf/soda/ChatterjeeJH04,DBLP:conf/fossacs/Zielonka04}. The sets of all strategies of Even and Odd are denoted by $\EvenStrategies$ and $\OddStrategies$, respectively.

A \emph{pair of strategies} is $\strategies = (\EvenStrategy,\OddStrategy) \in \EvenStrategies \times \OddStrategies$. For notational convenience, we also refer to $\strategies = (\EvenStrategy,\OddStrategy)$ as the function $\strategies : S \to A$ such that $\strategies(s) = \EvenStrategy(s)$ if $s \in \EvenState$ and $\strategies(s) = \OddStrategy(s)$ if $s \in \OddState$.

Given an SPG $\game$, fixing a pair $\strategies$ of strategies results in a Markov chain $\game^\strategies$ with a transition function $\delta^\strategies : S \to \dist(S)$ where $\delta^\strategies(s,s') = \delta(s,\strategies(s),s')$. Under $\strategies$, actions $a_i$ within a path can be omitted.

The probability $\mathbb{P}^\strategies$ of a finite path $\rho = s_0 s_1 \ldots s_n \in S^*$ is given by $\mathbb{P}^\strategies(\rho) := \prod_{i \in [0,n-1]} \delta^\strategies(s_i,s_{i+1})$. As described in \cite{DBLP:books/daglib/0020348}, the probability can be extended to all measurable sets of $S^\omega$. For a starting state $s \in S$ and a subset $S' \subseteq S$, we denote the probability of eventually reaching $S'$ from $s$ and that of always staying in $S'$ by $\mathbb{P}^\strategies_s(\textsf{F}\,S')$ and $\mathbb{P}^\strategies_s(\textsf{G}\,S')$, respectively.

\myparskip{Parity Objective.} For an infinite path $\rho = s_0 a_0 s_1 a_1 \ldots \in (S \times A)^\omega$, the set of states which are visited infinitely often in $\rho$ is denoted by $\infi(\rho) := \{ s \in S : \forall i \in \mathbb{N}, \exists j > i, s = s_j \}$. For $S' \subseteq S$, let $\mathcal{P}(S') := \{ \mathcal{P}(s') : s' \in S' \}$, where $\mathcal{P} : S \to \mathbb{N}$ is the priority function from Definition~\ref{def:game}.

The \emph{(Even) parity objective} requires that the \emph{minimum} priority visited infinitely often along an infinite path is even. That is, we are interested in the measurable set $\parity_\sq := \{ \textnormal{infinite path } \rho : \min(\mathcal{P}(\infi(\rho))) \textnormal{ is even} \}$. The \emph{Odd parity objective} and the measurable set $\parity_\tri$ are defined analogously. Throughout the paper, when a player is not specified, the parity objective refers to the \emph{Even} parity objective.

The goal of this work is to determine $\mathbb{P}^\strategies_s(\parity_\sq)$ under a pair of optimal strategies $\strategies = (\EvenStrategy,\OddStrategy)$, formally defined by the \emph{value function} $\mathcal{V} : S \to [0,1]$ as follows (the equality is proved in \cite{DBLP:conf/fossacs/Zielonka04}).

\begin{mydefinition}[value function $\mathcal{V}$]\label{def:valfunc}
Let $\game$ be an SPG as in Definition~\ref{def:game}. The \emph{value function} is the function $\mathcal{V} : S \to [0,1]$ defined by
\begin{align*}
\mathcal{V}(s) := &\max_{\EvenStrategy \in \EvenStrategies} \min_{\OddStrategy \in \OddStrategies} \mathbb{P}^\strategies_s(\parity_\sq)\\
= &\min_{\OddStrategy \in \OddStrategies} \max_{\EvenStrategy \in \EvenStrategies} \mathbb{P}^\strategies_s(\parity_\sq).
\end{align*}
\end{mydefinition}

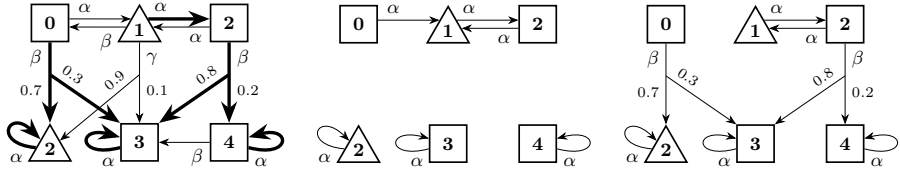
\begin{figure}[t]
    \centering
    \scalebox{.789}{\begin{tikzpicture}[
    box/.style={anchor=south, minimum size=25pt, regular polygon, regular polygon sides=4, inner sep=0pt},
    tri/.style={anchor=south, minimum size=23pt, regular polygon, regular polygon sides=3, inner sep=0pt}
]
    \node [state, tri] at (  0, 0) (s00) {$\textbf{2}$};
    \node [state, box] at (1.5, 0) (s01) {$\textbf{3}$};
    \node [state, box] at (  3, 0) (s02) {$\textbf{4}$};
    \node [state, box] at (  0, 2) (s10) {$\textbf{0}$};
    \node [state, tri] at (1.5, 2) (s11) {$\textbf{1}$};
    \node [state, box] at (  3, 2) (s12) {$\textbf{2}$};

    \path (s00) edge[out=175, in=135, looseness=12, ultra thick] node[below=1.5pt, pos=0.2]{$\alpha$} (s00);
    \path (s01) edge[loop left, looseness=11, ultra thick] node[below, pos=0.1]{$\alpha$} (s01);
    \path (s02) edge[out=345, in=15, looseness=11, ultra thick] node[below, pos=0.11]{$\alpha$} (s02);
    \path (s02) edge node[below, pos=0.25]{$\beta$} (s01);
    \path ($(s10.0)+(0,0.5ex)$) edge node[above, pos=0.25]{$\alpha$} ($(s11.150)+(0.3ex,0.5ex)$);
    \path ($(s11.150)-(0.3ex,0.5ex)$) edge node[below, pos=0.35]{$\beta$} ($(s10.0)-(0,0.5ex)$);
    \path ($(s11.30)+(-0.3ex,0.5ex)$) edge[ultra thick] node[above, pos=0.2]{$\alpha$} ($(s12.180)+(0,0.5ex)$);
    \path ($(s12.180)-(0,0.5ex)$) edge node[below, pos=0.25]{$\alpha$} ($(s11.30)-(-0.3ex,0.5ex)$);
    \path (s10) edge[ultra thick] node[left, pos=0.2]{$\beta$} coordinate[pos=0.4](m0) node[left, pos=0.6]{\scriptsize$0.7$} (s00);
    \path (m0) edge[ultra thick] node[above, sloped, pos=0.22]{\scriptsize$0.3$} (s01.135);
    \path (s11) edge node[right, pos=0.2]{$\gamma$} coordinate[pos=0.4](m1) node[right, pos=0.6]{\scriptsize$0.1$} (s01);
    \path (m1) edge node[above, sloped, pos=0.22]{\scriptsize$0.9$} (s00);
    \path (s12) edge[ultra thick] node[right, pos=0.2]{$\beta$} coordinate[pos=0.4](m2) node[right, pos=0.6]{\scriptsize$0.2$} (s02);
    \path (m2) edge[ultra thick] node[above, sloped, pos=0.22]{\scriptsize$0.8$} (s01.45);
\end{tikzpicture}\begin{tikzpicture}[
    box/.style={anchor=south, minimum size=25pt, regular polygon, regular polygon sides=4, inner sep=0pt},
    tri/.style={anchor=south, minimum size=23pt, regular polygon, regular polygon sides=3, inner sep=0pt}
]
    \node [state, tri] at (  0, 0) (s00) {$\textbf{2}$};
    \node [state, box] at (1.5, 0) (s01) {$\textbf{3}$};
    \node [state, box] at (  3, 0) (s02) {$\textbf{4}$};
    \node [state, box] at (  0, 2) (s10) {$\textbf{0}$};
    \node [state, tri] at (1.5, 2) (s11) {$\textbf{1}$};
    \node [state, box] at (  3, 2) (s12) {$\textbf{2}$};

    \path (s00) edge[out=175, in=135, looseness=12] node[below=1.5pt, pos=0.2]{$\alpha$} (s00);
    \path (s01) edge[loop left, looseness=11] node[below, pos=0.1]{$\alpha$} (s01);
    \path (s02) edge[out=345, in=15, looseness=11] node[below, pos=0.11]{$\alpha$} (s02);
    \path ($(s10.0)+(0,0.5ex)$) edge node[above, pos=0.25]{$\alpha$} ($(s11.150)+(0.3ex,0.5ex)$);
    \path ($(s11.30)+(-0.3ex,0.5ex)$) edge node[above, pos=0.2]{$\alpha$} ($(s12.180)+(0,0.5ex)$);
    \path ($(s12.180)-(0,0.5ex)$) edge node[below, pos=0.25]{$\alpha$} ($(s11.30)-(-0.3ex,0.5ex)$);
\end{tikzpicture}\begin{tikzpicture}[
    box/.style={anchor=south, minimum size=25pt, regular polygon, regular polygon sides=4, inner sep=0pt},
    tri/.style={anchor=south, minimum size=23pt, regular polygon, regular polygon sides=3, inner sep=0pt}
]
    \node [state, tri] at (  0, 0) (s00) {$\textbf{2}$};
    \node [state, box] at (1.5, 0) (s01) {$\textbf{3}$};
    \node [state, box] at (  3, 0) (s02) {$\textbf{4}$};
    \node [state, box] at (  0, 2) (s10) {$\textbf{0}$};
    \node [state, tri] at (1.5, 2) (s11) {$\textbf{1}$};
    \node [state, box] at (  3, 2) (s12) {$\textbf{2}$};

    \path (s00) edge[out=175, in=135, looseness=12] node[below=1.5pt, pos=0.2]{$\alpha$} (s00);
    \path (s01) edge[loop left, looseness=11] node[below, pos=0.1]{$\alpha$} (s01);
    \path (s02) edge[out=345, in=15, looseness=11] node[below, pos=0.11]{$\alpha$} (s02);
    \path ($(s11.30)+(-0.3ex,0.5ex)$) edge node[above, pos=0.2]{$\alpha$} ($(s12.180)+(0,0.5ex)$);
    \path ($(s12.180)-(0,0.5ex)$) edge node[below, pos=0.25]{$\alpha$} ($(s11.30)-(-0.3ex,0.5ex)$);
    \path (s10) edge node[left, pos=0.2]{$\beta$} coordinate[pos=0.4](m0) node[left, pos=0.6]{\scriptsize$0.7$} (s00);
    \path (m0) edge node[above, sloped, pos=0.22]{\scriptsize$0.3$} (s01.135);
    \path (s12) edge node[right, pos=0.2]{$\beta$} coordinate[pos=0.4](m2) node[right, pos=0.6]{\scriptsize$0.2$} (s02);
    \path (m2) edge node[above, sloped, pos=0.22]{\scriptsize$0.8$} (s01.45);
\end{tikzpicture}}
    \caption{Left: an SPG, where the priority $\mathcal{P}(\cdot)$ is shown within each state. The labels $1$ for $\delta(s,a,s') = 1$ are omitted. We denote $\sq_i$ and $\tri_i$ for Even's and Odd's state with priority $i$, respectively. Optimal strategies are in bold. Middle and right: the subgames under $\av'$ in iteration $i = 2$ and $3$, respectively, following an execution of Algorithm~\ref{alg:ec-spg-bvi}.}
    \label{fig:spg}
\end{figure}

\begin{example}
Figure~\ref{fig:spg} (left) illustrates an example of an SPG, where Even's states and Odd's states are shown as $\sq$ and $\tri$, respectively. We denote $\sq_i$ and $\tri_i$ for Even's and Odd's state with priority $i$, respectively. Optimal strategies are in bold. Its value function is $\mathcal{V} : \tri_2,\sq_4 \mapsto 1 \mid \sq_0 \mapsto 0.7 \mid \tri_1,\sq_2 \mapsto 0.2 \mid \sq_3 \mapsto 0$.
\end{example}

\mypar{End Component.} The notion of \emph{end component} \cite{DBLP:books/daglib/0020348} is first introduced in the context of MDPs and is then extended (e.g., as in \cite{DBLP:phd/us/Alfaro97,DBLP:conf/cav/KelmendiKKW18,DBLP:journals/iandc/EisentrautKKW22,DBLP:conf/lics/KretinskyMW23}) for stochastic games. We state the definition as follows.

\begin{mydefinition}[end component $E = (S_E,\textnormal{\av}_E)$]\label{def:ec}
Let $\game$ be an SPG as in Definition~\ref{def:game}. A pair $E = (S_E,\av_E)$, where $\emptyset \neq S_E \subseteq S$ and $\av_E : S_E \to 2^A \setminus \{\emptyset\}$ such that $\av_E(s) \subseteq \av(s)$ for all $s \in S_E$, is an \emph{end component (EC)} if
\begin{enumerate}[label=(\roman*)]
    \item $\forall s \in S_E, \forall a \in \av_E(s), \post(s,a) \subseteq S_E$, and
    \item $\forall s',s'' \in S_E$, there exists a finite path $\rho = s_0 a_0 s_1 a_1 \ldots s_n \in (S_E \times A)^* \times S_E$ for some $n \in \mathbb{N}$ such that $s_0 = s'$, $s_n = s''$, and $a_i \in \av_E(s_i)$ for all $0 \leq i < n$.
\end{enumerate}

Let $S' \subseteq S$, $\av' : S \to 2^A \setminus \{\emptyset\}$, and $E = (S_E, \av_E)$ be an EC. We say $E$ is \emph{within $(S',\av')$} if $S_E \subseteq S'$ and $\av_E(s) \subseteq \av'(s)$ for all $s \in S_E$. We say $E$ is \emph{maximal within $(S',\av')$} if $E$ is within $(S',\av')$ and there does not exist an EC $E' = (S_{E'}, \av_{E'})$ within $(S',\av')$ such that (i) $S_E \subsetneq S_{E'}$, or (ii) $\av_{E}(s) \subsetneq \av_{E'}(s)$ for some $s \in S_E$. We say $E$ is \emph{bottom within $(S',\av')$} if $E$ is within $(S',\av')$ and $\post(s,a) \subseteq S' \implies \post(s,a) \subseteq S_E$ for all $s \in S_E$ and $a \in \av'(s)$. We denote the set of all bottom maximal ECs (BMECs) within $(S',\av')$ by $\textsc{BMEC}(S',\av')$.
\end{mydefinition}

Given $S' \subseteq S$ and $\av' : S \to 2^A \setminus \{\emptyset\}$, the set $\textsc{BMEC}(S',\av')$ can be computed in polynomial time in the size of $S'$ and $\av'$ \cite{DBLP:books/daglib/0020348}.

\begin{example}
Consider the SPG in Figure~\ref{fig:spg} (left). A pair $E_1 = (S_{E_1},\av_{E_1})$, where $S_{E_1} = \{\sq_0,\tri_1\}$ and $\av_{E_1} : \sq_0 \mapsto \{\alpha\} \mid \tri_1 \mapsto \{\beta\}$, is an EC. The EC $E_1$ is \emph{not} maximal within $(S,\av)$ since there exists an EC $E_2 = (S_{E_2}, \av_{E_2})$, where $S_{E_2} = \{\sq_0,\tri_1,\sq_2\}$ and $\av_{E_2}: \sq_0 \mapsto \{\alpha\} \mid \tri_1 \mapsto \{\alpha,\beta\} \mid \sq_2 \mapsto \{\alpha\}$, and $S_{E_1} \subsetneq S_{E_2}$. The EC $E_2$ is maximal within $(S,\av)$. However, $E_2$ is not bottom within $(S,\av)$ since $\post(\sq_0,\beta) \not\subseteq S_{E_2}$. An EC $E_3 = (S_{E_3},\av_{E_3})$, where $S_{E_3} = \{\tri_2\}$ and $\av_{E_3}: \tri_2 \mapsto \{\alpha\}$, is bottom and maximal within $(S,\av)$. In contrast, a pair $E_4 = (S_{E_4}, \av_{E_4})$, where $S_{E_4} = \{\sq_3,\sq_4\}$ and $\av_{E_4} : \sq_3 \mapsto \{\alpha\} \mid \sq_4 \mapsto \{\alpha,\beta\}$, is \emph{not} an EC, as there is no finite path from $\sq_3$ to $\sq_4$ inside $E_4$.
\end{example}

\subsection{Fixed Points in a Complete Lattice}\label{subsec:lattice}

Various quantitative probabilistic model checking problems, such as computing reachability probabilities, can be formulated as computations of suitable fixed points in suitable complete lattices \cite{DBLP:books/daglib/0020348}. We provide brief backgrounds on complete lattices and relevant results.

For our setting, we are interested in the complete lattice $(L := [0,1]^S, \leq)$ where, for $f,f' \in L$, we write $f \leq f'$ if and only if $f(s) \leq f'(s)$ for all $s \in S$. Its least element is $\bot : S \to [0,1]$ with $\bot(s) = 0$ and its greatest element is $\top : S \to [0,1]$ with $\top(s) = 1$ for all $s \in S$.

A function $\lfunc : L \to L$ is \emph{monotone} if $f \leq f'$ implies $\lfunc f \leq \lfunc f'$. A lattice element $f \in L$ is a \emph{fixed point} of $\lfunc$ if $f = \lfunc f$. The \emph{least fixed point} of $\lfunc$, denoted by $\mu\lfunc$, is a fixed point of $\lfunc$ satisfying, for any $f \in L$, $f = \lfunc f \implies \mu\lfunc \leq f$. The \emph{greatest fixed point} of $\lfunc$, denoted by $\nu\lfunc$, is defined in a similar manner.

Below is a notable theorem regarding the least and greatest fixed points.

\begin{theorem}[(special case of) Kleene \cite{DBLP:journals/dm/Baranga91}]\label{thm:kleene}
Let $(L, \leq)$ be a complete lattice and $\lfunc : L \to L$ be an $\omega$-continuous function (i.e., $\lfunc(\sup L') = \sup\{ \lfunc f : f \in L' \}$ for any increasing $\omega$-chain $L' \subseteq L$). Then, the \emph{bottom-up Kleene sequence} $\bot \leq \lfunc\bot \leq \lfunc^2\bot \leq \cdots$ stabilizes at $\omega$ such that $\mu\lfunc = \lfunc^\omega\bot = \lfunc^{\omega+1}\bot = \cdots\;$.\footnote{For the bottom-up sequence, we define $\lfunc^\omega\bot := \sup\{\lfunc^i\bot : i \in \mathbb{N}\}$ and $\lfunc^{\omega+(i+1)}\bot := \lfunc(\lfunc^{\omega+i}\bot)$ for each $i \in \mathbb{N}$. Those for the top-down sequence are defined dually.}

Dually, let $\lfunc$ be an $\omega^\textnormal{op}$-cocontinuous function (i.e., $\lfunc(\inf L') = \inf\{ \lfunc f : f \in L' \}$ for any decreasing $\omega$-chain $L' \subseteq L$). Then, the \emph{top-down Kleene sequence} $\top \geq \lfunc\top \geq \lfunc^2\top \geq \cdots$ stabilizes at $\omega$ such that $\nu\lfunc = \lfunc^\omega\top = \lfunc^{\omega+1}\top = \cdots\;$.
\end{theorem}

\mypar{Value Iteration.} When a problem is formulated as the computation of the least fixed point $\mu\lfunc$ of $\lfunc$, one can apply Theorem~\ref{thm:kleene} to approximate this least fixed point: starting from $\bot$ and iteratively applying $\lfunc$ to the current lattice element. This procedure is commonly known as \emph{value iteration}. Note that $\lfunc^i\bot \leq \mu\lfunc$ holds for all $i \in \mathbb{N}$.

Since it could happen that $\lfunc^i\bot < \lfunc^\omega\bot = \mu\lfunc$ for all $i \in \mathbb{N}$, it may not be possible to precisely compute $\mu\lfunc$ via value iteration. Nevertheless, in practice, it is usually sufficient to compute an \emph{approximation} that is $\varepsilon$-close to $\mu\lfunc$ for some desired precision $\varepsilon > 0$. The rest of Section~\ref{subsec:lattice} targets at this problem.

A straightforward approach is to stop the iteration when $\max_{s \in S}((\lfunc^i\bot)(s)-(\lfunc^{i-1}\bot)(s)) \leq \varepsilon$. However, it is shown in \cite{DBLP:journals/tcs/HaddadM18} that such $\lfunc^i\bot$ is not guaranteed to be $\varepsilon$-close to $\mu\lfunc$. To our best knowledge, the bottom-up Kleene sequence alone is not sufficient to provide an $\varepsilon$-close approximation of $\mu\lfunc$.

\myparskip{Bounded Value Iteration (BVI).} When $\mu\lfunc = \nu\lfunc$ (i.e., when $\lfunc$ has a \emph{unique} fixed point), one can perform value iteration from both $\bot$ and $\top$. This procedure is known as \emph{bounded value iteration} (also called \emph{interval iteration}). Note that $\lfunc^i\top \geq \nu\lfunc$ holds for all $i \in \mathbb{N}$. During this bounded iteration, when $\max_{s \in S}((\lfunc^i\top)(s)-(\lfunc^i\bot)(s)) \leq 2\varepsilon$, one can guarantee that $\frac{1}{2}(\lfunc^i\bot+\lfunc^i\top)$ is $\varepsilon$-close to $\mu\lfunc = \nu\lfunc$.

We point out that not all functions $\lfunc$ have a unique fixed point. In the case where we would like to compute $\mu\lfunc$ (or $\nu\lfunc$), it may be possible to construct $\lfunc'$ with a unique fixed point such that $\mu\lfunc'=\nu\lfunc'=\mu\lfunc$ (or $\nu\lfunc$), and when possible, the construction may not be straightforward. Interested readers are referred to, e.g., \cite{DBLP:phd/us/Alfaro97,DBLP:books/daglib/0020348,DBLP:conf/atva/PhalakarnTH25} for further details.

\subsection{Bellman Operator}\label{subsec:bellman}

Building on Section~\ref{subsec:lattice}, we now review the \emph{Bellman operator}, a key function used in computing value functions for a range of quantitative objectives. Our definition is based on a notion of \emph{state-action expectation} as follows.

\begin{mydefinition}[state-action expectation $\phi_f(s,a)$]\label{def:sa-exp}
Let $\game$ be an SPG as in Definition~\ref{def:game}. The \emph{state-action expectation} of a function $f \in [0,1]^S$ for a state-action pair $(s,a) \in S \times A$ is $\phi_f(s,a) := \sum_{s' \in S} \delta(s,a,s')\cdot f(s')$.
\end{mydefinition}

\begin{mydefinition}[Bellman operator $\bellman$]\label{def:bellman}
Let $\game$ be an SPG as in Definition~\ref{def:game}. The \emph{Bellman operator} is the function $\bellman : [0,1]^S \to [0,1]^S$ where, for $f \in [0,1]^S$ and $s \in S$,
$$
(\bellman f)(s) :=
\begin{cases}
    \textstyle\max_{a \in \av(s)} \phi_f(s,a) & \textnormal{if } s \in \EvenState \textnormal{, and}\\
    \textstyle\makebox[\widthof{$\max_{a \in \av(s)} \phi_f(s,a)$}][r]{$\min_{a \in \av(s)} \phi_f(s,a)$} & \textnormal{if } s \in \OddState.
\end{cases}
$$
\end{mydefinition}

It is well-known that (i) $\bellman$ is monotone, and (ii) the value function $\mathcal{V}$ (cf.~Definition~\ref{def:valfunc}) is a fixed point of the Bellman operator $\bellman$ (i.e., $\mathcal{V} = \bellman\mathcal{V}$). However, in general, $\mathcal{V}$ does not coincide with either the least or the greatest fixed point of $\bellman$ (i.e., $\mathcal{V} \neq \mu\bellman$ and $\mathcal{V} \neq \nu\bellman$).

For the reachability and safety objectives, we consider the following modification of the Bellman operator.

\begin{mydefinition}[modified Bellman operator $\bellman_{(S',c)}$]\label{def:mod-bellman}
Let $\game$ be an SPG as in Definition~\ref{def:game}, $S' \subseteq S$, and $c \in [0,1]$. The \emph{modified Bellman operator} is the function $\bellman_{(S',c)} : [0,1]^S \to [0,1]^S$ where, for $f \in [0,1]^S$ and $s \in S$,
$$
(\bellman_{(S',c)} f)(s) :=
\begin{cases}
    c & \textnormal{if } s \in S' \textnormal{,}\\
    \textstyle\max_{a \in \av(s)} \phi_f(s,a) & \textnormal{if } s \in \EvenState \setminus S' \textnormal{, and}\\
    \textstyle\makebox[\widthof{$\max_{a \in \av(s)} \phi_f(s,a)$}][r]{$\min_{a \in \av(s)} \phi_f(s,a)$} & \textnormal{if } s \in \OddState \setminus S'.
\end{cases}
$$
\end{mydefinition}

The \emph{optimal} reachability probability $\mathbb{P}^\strategies_s(\mathsf{F}\,S')$ under an optimal pair of strategies $\strategies$ is easily shown to be the least fixed point of $\bellman_{(S',1)}$, i.e., $\mathbb{P}^\strategies_s(\mathsf{F}\,S') = \mu\bellman_{(S',1)}$. Dually, the \emph{optimal} safety probability $\mathbb{P}^\strategies_s(\mathsf{G}\,S')$ under an optimal pair of strategies $\strategies$ is easily shown to be the greatest fixed point of $\bellman_{(S \setminus S',0)}$, i.e., $\mathbb{P}^\strategies_s(\mathsf{G}\,S') = \nu\bellman_{(S \setminus S',0)}$.

\section{Related Work}\label{sec:related}

\subsection{Solving Qualitative SPGs}\label{subsec:qual-prob}

As our algorithm for solving the quantitative SPG problem relies on solving its qualitative counterpart, we review the relevant literature below.

\begin{mydefinition}[winning sets $W_\sq,W_\tri$]\label{def:winset}
Let $\game$ be an SPG as in Definition~\ref{def:game} and $\mathcal{V}$ be the value function as in Definition~\ref{def:valfunc}. The \emph{(almost-sure qualitative) winning set of Even} is the set $W_\sq := \{ s \in S : \mathcal{V}(s) = 1 \}$. The \emph{(almost-sure qualitative) winning set of Odd} is the set $W_\tri := \{ s \in S : \mathcal{V}(s) = 0 \}$.
\end{mydefinition}

\begin{myproblem}[qualitative SPG problem]\label{prob:qual-prob}
Given an SPG $\game$, the \emph{qualitative SPG problem} is to output the winning set $W_\sq$ of Even.
\end{myproblem}

Note that the winning set $W_\tri$ of Odd for an SPG $\game$ is the same as the winning set $\tilde{W}_\sq$ of Even for the SPG $\tilde{\game}$ where $\tilde{\mathcal{P}}(s) := \mathcal{P}(s)+1$ for each $s \in S$.

There are two approaches to solve the qualitative SPG problem: the nested fixed point algorithm \cite{DBLP:conf/cav/HahnSTZ16} and the gadget construction \cite{DBLP:conf/csl/ChatterjeeJH03}.

\myparskip{Nested Fixed Point Algorithm.} The first approach is a recursive-style algorithm of Hahn et al. \cite{DBLP:conf/cav/HahnSTZ16}. It can be considered as an adaptation of \cite{DBLP:journals/apal/McNaughton93,DBLP:journals/tcs/Zielonka98} for non-stochastic parity games (i.e., with a transition function $\delta : S \times A \to S$).

We first define \emph{attractors}; they can be computed via iterative fixed-point algorithms.

\begin{mydefinition}[attractor $\textnormal{\attr}$]\label{def:attr}
Let $\game$ be an SPG as in Definition~\ref{def:game}. The \emph{attractors} are the functions $\attr_{\sq,\bowtie}$ and $\attr_{\tri,\bowtie}: 2^S \to 2^S$, where $\bowtie\,\,\in \{{>0},{=1}\}$, defined by
\begin{align*}
\attr_{\sq,\bowtie}(S') &:= \{ s \in S : \exists \EvenStrategy \in \EvenStrategies, \forall \OddStrategy \in \OddStrategies, \mathbb{P}^{\strategies}_s(\mathsf{F}\,S') \bowtie\},\\
\attr_{\tri,\bowtie}(S') &:= \{ s \in S : \exists \OddStrategy \in \OddStrategies, \forall \EvenStrategy \in \EvenStrategies, \mathbb{P}^{\strategies}_s(\mathsf{F}\,S') \bowtie\}.
\end{align*}
\end{mydefinition}

\begin{algorithm}[t]
\caption{the algorithm of \cite{DBLP:conf/cav/HahnSTZ16} solving the qualitatve SPG problem.}\label{alg:qualitative}
\DontPrintSemicolon
\SetAlgoNoLine
$\textsc{QualitativeSPG}(\game)$\\\Indp
    \lIf{$S = \emptyset$}{\Return $(\emptyset, \emptyset)$}
    $x \gets \min(\mathcal{P}(S)); X \gets \{ s \in S : \mathcal{P}(s) = x \}$\\
    \If{$x$ \textup{is even}}{
        $(W'_\sq,L'_\sq) \gets \textsc{QualitativeSPG}(\game \setminus \attr_{\sq,>0}(X))$\\
        \lIf{$L'_\sq = \emptyset$}{\Return $(S, \emptyset)$}
        $(W''_\sq, L''_\sq) \gets \textsc{QualitativeSPG}(\game \setminus \attr_{\tri,>0}(L'_\sq))$\\
        \Return $(W''_\sq, L'_\sq \cup L''_\sq)$
    }
    \If{$x$ \textup{is odd}}{
        $(W'_\sq,L'_\sq) \gets \textsc{QualitativeSPG}(\game \setminus \attr_{\tri,>0}(X))$\\
        \lIf{$W'_\sq = \emptyset$}{\Return $(\emptyset, S)$}
        $(W''_\sq, L''_\sq) \gets \textsc{QualitativeSPG}(\game \setminus \attr_{\sq,=1}(W'_\sq))$\\
        \Return $(W'_\sq \cup W''_\sq, L''_\sq)$
    }\Indm
\;
$\textsc{EvenWinSet}(\game)$\\\Indp
    $(W_\sq,L_\sq) \gets \textsc{QualitativeSPG}(\game)$; \Return $W_\sq$\\\Indm
\;
$\textsc{OddWinSet}(\game)$\\\Indp
    \Return $\textsc{EvenWinSet}(\tilde{\game})$ where $\tilde{\game}$ has $\tilde{\mathcal{P}}(s) := \mathcal{P}(s) + 1$ for each $s \in S$
\end{algorithm}

The algorithm \textsc{QualitativeSPG} in Algorithm~\ref{alg:qualitative} outputs $(W_\sq,L_\sq)$ such that $W_\sq$ is the winning set of Even (cf.~Definition~\ref{def:winset}) and $L_\sq = \{ s \in S : \mathcal{V}(s) < 1 \} = S \setminus W_\sq$. We let $\game \setminus X$ denote the subgame $\game' = (S',\, S_\sq \setminus X,\, S_\tri \setminus X,A,\av',\delta,\mathcal{P})$ where $S' := S \setminus X$ and $\av' : S' \to 2^A \setminus \{\emptyset\}$ with $\av'(s') := \{ a' \in \av(s') : \post(s',a') \subseteq S' \}$ for all $s' \in S'$.




\myparskip{Gadget Construction.} The second approach is to translate an SPG into a non-stochastic parity game via the \emph{gadget construction} by Chatterjee~et~al.~\cite{DBLP:conf/csl/ChatterjeeJH03}. The construction is guaranteed to preserve the winning set of one player.

After the construction, any algorithm solving non-stochastic parity games applies, e.g., \cite{DBLP:journals/siamcomp/CaludeJKLS22,DBLP:conf/atva/FriedmannL09,DBLP:conf/lics/JurdzinskiL17,DBLP:journals/siamcomp/JurdzinskiPZ08,DBLP:conf/lics/Lehtinen18,DBLP:journals/apal/McNaughton93,DBLP:conf/mfcs/Parys19,DBLP:journals/tcs/Zielonka98}. Nonetheless, the size of the constructed non-stochastic parity game can increase by a factor of $\mathcal{O}(\max(\mathcal{P}(S)))$ compared to the original SPG. This can lead to a longer solving time as observed in \cite{DBLP:conf/cav/HahnSTZ16}.

\subsection{Solving Quantitative SPGs}\label{subsec:quan-prob}

Below is the definition of the problem, which is the focus of this work.

\begin{mdframed}
\begin{myproblem}[quantitative SPG problem]\label{prob:quant-prob}
Given an SPG $\game$, an initial state $\init \in S$, and a precision $\varepsilon > 0$, the \emph{quantitative SPG problem} is to output a value that is $\varepsilon$-close to $\mathcal{V}(\init)$ (i.e., whose difference from $\mathcal{V}(\init)$ is no more than $\varepsilon$).
\end{myproblem}
\end{mdframed}

To the best of our knowledge, the only available approach to solve the quantitative SPG problem is by \emph{strategy iteration} (also called \emph{strategy improvement}) \cite{DBLP:conf/stacs/ChatterjeeH06,DBLP:conf/vmcai/HahnST017}. The algorithm of \cite{DBLP:conf/vmcai/HahnST017} was implemented and is available as part of the EPMC model checking tool \cite{DBLP:conf/vmcai/FuHLSSTZ22}.

Roughly speaking, strategy iteration begins with an arbitrary strategy $\sigma_{\sq,0}$ of Even. By fixing $\sigma_{\sq,i}$ for $i \in \mathbb{N}$, the SPG becomes an MDP, where one can compute an optimal strategy $\sigma_{\tri,i}$ of Odd for the parity objective \cite{DBLP:phd/us/Alfaro97,DBLP:journals/tac/CourcoubetisY98}. By fixing $\sigma_{\tri,i}$, one may \emph{improve} $\sigma_{\sq,i}$ to $\sigma_{\sq,i+1}$ such that the winning probabilities of Even increase. When $\sigma_{\sq,i}$ cannot be further improved, $\sigma_{\sq,i}$ is optimal for Even. The value $\mathcal{V}(\init)$ can then be computed accordingly.

It is mentioned in \cite{DBLP:conf/vmcai/HahnST017} that computing an optimal strategy for a parity objective of an MDP can become a bottleneck. Their experiment showed that more than half of the computation time was contributed to MDP solving.

\section{Our BVI Algorithm for SPG}\label{sec:alg}

\subsection{Intuition}\label{subsec:intuition}

Recall that the parity objective is determined by the set of states visited infinitely often. Any infinite path must eventually remain within some end component (EC) in order to visit a state infinitely often. We begin by defining the following classifications of ECs.

\begin{mydefinition}[Even-dominating and Odd-dominating ECs]\label{def:win-EC}
Let $\game$ be an SPG as in Definition~\ref{def:game} and $E = (S_E,\av_E)$ be an EC (cf.~Definition~\ref{def:ec}). We say that $E$ is \emph{Even-dominating} if there exists a strategy $\EvenStrategy \in \EvenStrategies$ such that (i) $\EvenStrategy(s) \in \av_E(s)$ for all $s \in S_E \cap S_\sq$, and (ii) for all strategies $\OddStrategy \in \OddStrategies$ such that $\OddStrategy(s) \in \av_E(s)$ holds for all $s \in S_E \cap S_\tri$, we have $\mathbb{P}^{\strategies}_s(\parity_\sq) = 1$ for all $s \in S_E$. An \emph{Odd-dominating} EC is defined analogously.
\end{mydefinition}

Intuitively, an EC $E$ is Even-dominating if Even can ensure winning as long as both players remain within $E$. Since Odd loses by remaining within $E$, it is optimal for Odd to exit $E$ (if possible). Similarly, $E$ is Odd-dominating if Odd can ensure winning under the same condition, and it is optimal for Even to exit $E$ (if possible).

Based on the above intuition, our algorithm considers Even-dominating and Odd-dominating ECs, and adjusts the bounds for the winning probabilities accordingly. An example is provided in Section~\ref{subsec:example}.

\subsection{Algorithm}\label{subsec:alg}

Here is our BVI algorithm for the quantitative SPG problem.

\begin{algorithm}[t]
\caption{our SPG-BVI algorithm.}\label{alg:ec-spg-bvi}
\DontPrintSemicolon
\SetAlgoNoLine
$\textsc{SPG-BVI}(\game,\init,\varepsilon)$\\\Indp
    $\ell_0 \gets \bot$; $u_0 \gets \top$; $i \gets 0$\\
    \While{$u_i(\init) - \ell_i(\init) > 2\varepsilon$}{
        $i \gets i+1; \ell_i \gets \bellman \ell_{i-1}; u_i \gets \bellman u_{i-1}$\tcp*[f]{Definition \ref{def:bellman}}\label{alg2:line4}\\
        Let $\av' : S \to 2^A \setminus \{\emptyset\}$ be a function where $\av'(s) := \arg\max_{a \in \av(s)} \phi_{u_i}(s,a)$ for $s \in \EvenState$ and $\av'(s) := \,\arg\min_{a \in \av(s)} \,\phi_{\ell_i}(s,a)$ for $s \in \OddState$\label{alg2:line5}\\
        $\game' \gets (S,\EvenState,\OddState,A,\av',\delta,\mathcal{P})$\\
        $W'_\sq\;\! \gets \textsc{EvenWinSet}(\game')$; $W'_\tri \gets \textsc{OddWinSet}(\game')$\tcp*[f]{Algorithm \ref{alg:qualitative}}\\
        \ForEach(\tcp*[f]{Definition \ref{def:ec}}){$E = (S_E,\av_E) \in \textnormal{\textsc{BMEC}}(W'_\sq,\av')$}{
            $\ell_i(s) \gets \max(\ell_i(s),\exit_\tri(E,\ell_i))$ for each $s \in S_E$\label{alg2:line9}}
        \ForEach{$E = (S_E,\av_E) \in \textnormal{\textsc{BMEC}}(W'_\tri,\av')$}{
            $u_i(s) \gets \min(u_i(s),\exit_\sq(E,u_i))$ for each $s \in S_E$}
    }
    \Return $\frac{1}{2}(\ell_i(\init) + u_i(\init))$
\end{algorithm}

\begin{mdframed}
\begin{mydefinition}[$\textnormal{\textsc{SPG-BVI}}(\game,\init,\varepsilon)$]
Our SPG-BVI algorithm for solving the quantitative SPG problem is in Algorithm~\ref{alg:ec-spg-bvi}.
\end{mydefinition}
\end{mdframed}

The algorithm maintains two sequences of functions $\ell_i$ and $u_i : S \to [0,1]$ for $i \in \mathbb{N}$. They act as the bottom-up and top-down sequences for the value function $\mathcal{V}$, respectively. The first elements are initialized as $\ell_0 \gets \bot$ and $u_0 \gets \top$.

In the $i$-th iteration ($i > 0$), two new elements $\ell_i \gets \bellman \ell_{i-1}$ and $u_i \gets \bellman u_{i-1}$ are generated. After that, the algorithm computes the restricted actions for both players. For $s \in \EvenState$, an action $a \in \av(s)$ is available in $\av'(s)$ if $\phi_{u_i}(s,a)$ is \emph{maximum} among all actions in $\av(s)$. For $s \in \OddState$, an action $a \in \av(s)$ is available in $\av'(s)$ if $\phi_{\ell_i}(s,a)$ is \emph{minimum} among all actions in $\av(s)$. Such restriction induces the \emph{estimated-optimal subgame} $\game'$.

Next, the algorithm computes the winning sets $W'_\sq, W'_\tri$ of Even and Odd (cf.~Definition~\ref{def:winset}), respectively, for the estimated-optimal subgame $\game'$. We then look at each bottom maximal EC (BMEC) within $(W'_\sq,\av')$ (cf.~Definition~\ref{def:ec}), i.e., each $E \in \textsc{BMEC}(W'_\sq,\av')$. We show that such $E$ is \emph{Even-dominating} (cf.~Definition~\ref{def:win-EC}) in both $\game'$ and $\game$.

\begin{lemma}\label{lem:bmec-win}
Let $\game$ be an SPG as in Definition~\ref{def:game}, $\av' : S \to 2^A \setminus \{\emptyset\}$ such that $\av'(s) \subseteq \av(s)$ for all $s \in S$, $\game' = (S,S_\sq,S_\tri,A,\av',\delta,\mathcal{P})$ be a subgame of $\game$, $W'_\sq$ is the winning set of Even (cf.~Definition~\ref{def:winset}) for $\game'$, and $E \in \textnormal{\textsc{BMEC}}(W'_\sq,\av')$ (cf.~Definition~\ref{def:ec}). Then, $E$ is an Even-dominating EC (cf.~Definition~\ref{def:win-EC}) in both $\game'$ and $\game$.
\end{lemma}
\begin{proof}
Consider an EC $E = (S_E,\av_E) \in \textsc{BMEC}(W'_\sq,\av')$. Because $E$ is bottom within $(W'_\sq,\av')$, we obtain $\post(s,a) \subseteq W'_\sq \implies \post(s,a) \subseteq S_E$ for all $s \in S_E$ and $a \in \av'(s)$. Due to the maximality of $E$ within $(W'_\sq,\av')$, it must be the case that $\av_E$ includes all actions within $S_E$, i.e., $\post(s,a) \subseteq S_E \implies a \in \av_E(s)$ for all $s \in S_E$ and $a \in \av'(s)$.

Since $S_E \subseteq W'_\sq$, any (almost-sure qualitative) winning strategy $\sigma'_\sq : S \to A$ (with $\sigma'_\sq(s) \in \av'(s)$ for $s \in \EvenState$) of Even with respect to $\game'$ (which induces $W'_\sq$) gives $\post(s,\sigma'_\sq(s)) \subseteq W'_\sq$ for all $s \in S_E \cap \EvenState$. Also, Odd cannot escape $W'_\sq$ (i.e., $\post(s,\sigma'_\tri(s)) \subseteq W'_\sq$ for all $s \in S_E \cap \OddState$ and $\sigma'_\tri : S \to A$ with $\sigma'_\tri(s) \in \av'(s)$ for $s \in \OddState$). Notice that, for $\strategies' = (\sigma'_\sq,\sigma'_\tri)$, we get $\post(s,\strategies'(s)) \subseteq W'_\sq$ for all $s \in S_E$. This implies $\post(s,\strategies'(s)) \subseteq S_E$ and $\strategies'(s) \in \av_E(s)$ for all $s \in S_E$. Hence, such $\sigma'_\sq$ gives a strategy such that $E$ is Even-dominating. Thus, $E$ is Even-dominating in $\game'$. As all ECs in $\game'$ are ECs in $\game$, therefore $E$ is also Even-dominating in $\game$. \qed
\end{proof}

For each $E = (S_E,\av_E) \in \textsc{BMEC}(W'_\sq,\av')$, we set the value $\ell_i(s)$ for each $s \in S_E$ to be the maximum of $\ell_i(s)$ and the \emph{exit value} $\exit_\tri(E,\ell_i)$ from $E$ for Odd, defined below.

\begin{mydefinition}[exit value $\textnormal{\exit}$]\label{def:exit}
Let $\game$ be an SPG as in Definition~\ref{def:game}, $E = (S_E,\av_E)$ be an EC as in Definition~\ref{def:ec}, and $f : S \to [0,1]$ be a function. The \emph{exit values} $\exit_\sq$ and $\exit_\tri$ with respect to $E$ and $f$ are defined as
\begin{align*}
\exit_\sq(E,f) &:= \max \{ \phi_f(s,a) : s \in S_E \cap S_\sq, a \in \av(s) \setminus \av_E(s) \},\\
\exit_\tri(E,f) &:= \,\min \{ \phi_f(s,a) : s \in S_E \cap S_\tri, a \in \av(s) \setminus \av_E(s) \}
\end{align*}
with the convention that $\max \emptyset = 0$ and $\min \emptyset = 1$.
\end{mydefinition}

Dually, for each $E = (S_E,\av_E) \in \textsc{BMEC}(W'_\tri,\av')$, it is \emph{Odd-dominating} in both $\game'$ and $\game$. We set the value $u_i(s)$ for $s \in S_E$ to be $\min(u_i(s), \exit_\sq(E,u_i))$. The while loop repeats until $\ell_i(\init)$ and $u_i(\init)$ are close within $2\varepsilon$, at which point the algorithm stops and returns the middle point $\frac{1}{2}(\ell_i(\init)+u_i(\init))$.

\subsection{Example}\label{subsec:example}

We present a detailed example of executing \textsc{SPG-BVI} (Algorithm~\ref{alg:ec-spg-bvi}) with the SPG from Figure~\ref{fig:spg} (left).

\begin{example}
Let $\game$ be the SPG from Figure~\ref{fig:spg} (left). We represent a function $f : S \to [0,1]$ by a tuple, namely $f = (f(\sq_0),f(\tri_1),f(\sq_2),f(\tri_2),f(\sq_3),f(\sq_4))$. The algorithm begins with $\ell_0 = (0,0,0,0,0,0)$ and $u_0 = (1,1,1,1,1,1)$.

For $i = 1$, $\ell_1 \gets \bellman \ell_0 = \ell_0$ and $u_1 \gets \bellman u_0 = u_0$. Under $\ell_1$ and $u_1$, we have $\av' = \av$ (i.e., no actions are removed) and $\game' = \game$. The winning sets of both players are $W'_\sq = \{\tri_2,\sq_4\}$ and $W'_\tri = \{\sq_3\}$.

The first Even-dominating BMEC is $E = (S_E, \av_E)$ with $S_E = \{\tri_2\}$ and $\av_E : \tri_2 \mapsto \{\alpha\}$. The exit value from $E$ is $\exit_\tri(E,\ell_1) = 1$ (e.g., the set of exiting actions is empty). Thus, the algorithm sets $\ell_1(\tri_2) \gets 1$. Similarly, we have $\ell_1(\sq_4) \gets 1$ and $u_1(\sq_3) \gets 0$. At the end of iteration $i = 1$, we have $\ell_1 = (0,0,0,1,0,1)$ and $u_1 = (1,1,1,1,0,1)$.

For $i = 2$, $\ell_2 \gets \bellman \ell_1 = (0.7,0,0.2,1,0,1)$ and $u_2 \gets \bellman u_1 = (1,0.9,1,1,0,1)$. The subgame $\game'$ and its actions under $\av'$ are shown in Figure~\ref{fig:spg} (middle). The winning sets of both players are $W'_\sq = \{\tri_2,\sq_4\}$ and $W'_\tri = \{\sq_0,\tri_1,\sq_2,\sq_3\}$.

The BMECs involving $\tri_2,\sq_3,\sq_4$ do not result in any modifications of $\ell_2$ and $u_2$. When considering $E = (S_E,\av_E) \in \textsc{BMEC}(W'_\tri,\av')$ with $S_E = \{\tri_1,\sq_2\}$ and $\av_E : \tri_1 \mapsto \{\alpha\} \mid \sq_2 \mapsto \{\alpha\}$, it is Odd-dominating and Even must escape via $\beta \in \av(\sq_2) \setminus \av'(\sq_2)$. The exit value from $E$ is $\exit_\sq(E,u_2) = 0.2$, hence we obtain $u_2(\tri_1) \gets 0.2$ and $u_2(\sq_2) \gets 0.2$. At the end of iteration $i = 2$, we have $\ell_2 = (0.7,0,0.2,1,0,1)$ and $u_2 = (1,0.2,0.2,1,0,1)$.

For $i = 3$, the algorithm computes $\ell_3 \gets \bellman \ell_2 = (0.7,0.2,0.2,1,0,1)$ and $u_3 \gets \bellman u_2 = (0.7,0.2,0.2,1,0,1)$. The subgame $\game'$ and its actions under $\av'$ are shown in Figure~\ref{fig:spg} (right). The winning sets of both players are $W'_\sq = \{\tri_2,\sq_4\}$ and $W'_\tri = \{\sq_3\}$. No further adjustments are made on both $\ell_3$ and $u_3$. At the end of iteration $i = 3$, we have $\ell_3 = u_3$. Therefore, the while loop terminates.
\end{example}

\subsection{Correctness and Convergence}\label{subsec:proofs}

We prove the correctness and convergence of our \textsc{SPG-BVI} algorithm (Algorithm~\ref{alg:ec-spg-bvi}) by showing that (i) $\ell_i \leq \ell_{i+1} \leq \mathcal{V} \leq u_{i+1} \leq u_i$ for all $i \in \mathbb{N}$, and (ii) $\lim_{i \to \infty} \ell_i = \mathcal{V} = \lim_{i \to \infty} u_i$. We only provide the proofs for the sequence $\ell_i$. Similar arguments can be done for the sequence $u_i$.

\begin{lemma}\label{lemma:ec-monotone}
Assume the setting of Algorithm~\ref{alg:ec-spg-bvi}, $E = (S_E,\av_E)$ be an EC, and $f : S \to [0,1]$ be a function. Let $f'(s) := \max(f(s),\exit_\tri(E,f))$ if $s \in S_E$ and $f'(s) := f(s)$ otherwise. If $f \leq \bellman f$, then $f' \leq \bellman f'$.
\end{lemma}
\begin{proof}
Suppose $f \leq \bellman f$. By the definition of $f'$, we have $f \leq f'$. This implies $\bellman f \leq \bellman f'$ by the monotonicity of $\bellman$. If $f'(s) = f(s)$, then $f'(s) = f(s) \leq (\bellman f)(s) \leq (\bellman f')(s)$. The remaining case is that $s \in S_E$ and $f'(s) = \exit_\tri(E,f)$. We distinguish cases whether $s$ belongs to Even or Odd.

If $s \in S_E \cap S_\sq$, then there is an action $\hat{a} \in \av_E(s)$ where $\post(s,\hat{a}) \subseteq S_E$ (since $E$ is an EC). Consider an arbitrary $\hat{a}$-successor $s'$, that is, $s' \in \post(s,\hat{a})$, we have $s' \in S_E$ and $f'(s') \geq \exit_\tri(E,f)$ by the definition of $f'$. Hence, $\phi_{f'}(s,\hat{a}) \geq \exit_\tri(E,f)$. Therefore, $$f'(s) = \exit_\tri(E,f) \leq \phi_{f'}(s,\hat{a}) \leq \max_{a \in \av(s)} \phi_{f'}(s,a) = (\bellman f')(s).$$

If $s \in S_E \cap S_\tri$, then we consider each $a \in \av(s)$ and distinguish whether $a \in \av_E(s)$ or not. If $a \in \av_E(s)$, we have $\phi_{f'}(s,a) \geq \exit_\tri(E,f)$ as for $\hat{a} \in \av_E(s)$ above. Otherwise, $a \in \av(s) \setminus \av_E(s)$. By the definition of the exit value (Definition~\ref{def:exit}), $\phi_f(s,a) \geq \exit_\tri(E,f)$. Since $f' \geq f$, we have $\phi_{f'}(s,a) \geq \phi_f(s,a) \geq \exit_\tri(E,f)$. From both cases, we have $\phi_{f'}(s,a) \geq \exit_\tri(E,f)$ for any $a \in \av(s)$. So,
\begin{align*}
f'(s) = \exit_\tri(E,f) \leq \min_{a \in \av(s)} \phi_{f'}(s,a) = (\bellman f')(s).\tag*{\qed}
\end{align*}
\end{proof}

\begin{proposition}\label{prop:bellman-monotone}
In the setting of Algorithm~\ref{alg:ec-spg-bvi}, $\ell_i \leq \bellman\ell_i$ for all $i \in \mathbb{N}$.
\end{proposition}
\begin{proof}
We prove by induction on $i$. The base case $i = 0$ is trivially true as $\ell_0 = \bot$. For the step case, we assume that the statement is true for $i-1$: $\ell_{i-1} \leq \bellman\ell_{i-1}$, and show that it is true for $i$.

By the monotonicity of $\bellman$ and the induction hypothesis, we have $\bellman\ell_{i-1} \leq \bellman(\bellman\ell_{i-1})$. As Line~\ref{alg2:line4} of Algorithm~\ref{alg:ec-spg-bvi} performs $\ell_i \gets \bellman\ell_{i-1}$, this implies $\ell_i \leq \bellman\ell_i$ after Line~\ref{alg2:line4}. Moreover, by Lemma~\ref{lemma:ec-monotone}, the update of $\ell_i$ in Line~\ref{alg2:line9} preserves such property. Thus, $\ell_i \leq \bellman\ell_i$ holds at the end of the iteration. \qed
\end{proof}

\begin{proposition}\label{prop:increasing}
In the setting of Algorithm~\ref{alg:ec-spg-bvi}, $\ell_i \leq \ell_{i+1}$ for all $i \in \mathbb{N}$.
\end{proposition}
\begin{proof}
We instead show $\ell_{i-1} \leq \ell_i$ for $i > 0$. First, Line~\ref{alg2:line4} of Algorithm~\ref{alg:ec-spg-bvi} performs $\ell_i \gets \bellman\ell_{i-1}$. By Proposition~\ref{prop:bellman-monotone}, $\ell_{i-1} \leq \bellman\ell_{i-1} = \ell_i$. Moreover, the update of $\ell_i$ in Line~\ref{alg2:line9} is non-decreasing. Hence, $\ell_{i-1} \leq \ell_i$ holds at the end of the iteration. \qed
\end{proof}

\begin{lemma}\label{lem:below-V}
In the setting of Algorithm~\ref{alg:ec-spg-bvi}, $\ell_i \leq \mathcal{V}$ for all $i \in \mathbb{N}$.
\end{lemma}
\begin{proof}
We prove by induction on $i$. The base case $i = 0$ is trivially true as $\ell_0 = \bot$. For the step case, we assume that the statement is true for $i-1$: $\ell_{i-1} \leq \mathcal{V}$, and show that it is true for $i$.

By the monotonicity of $\bellman$ and the induction hypothesis, we have $\bellman \ell_{i-1} \leq \bellman \mathcal{V} = \mathcal{V}$ (cf.~Definition~\ref{def:bellman}). As Line~\ref{alg2:line4} of Algorithm~\ref{alg:ec-spg-bvi} performs $\ell_i \gets \bellman\ell_{i-1}$, this implies $\ell_i \leq \mathcal{V}$ after Line~\ref{alg2:line4}. The rest is to show that the update of $\ell_i$ in Line~\ref{alg2:line9} preserves such property.

Consider an EC $E = (S_E,\av_E) \in \textsc{BMEC}(W'_\sq,\av')$. By Lemma~\ref{lem:bmec-win}, $E$ is Even-dominating in $\game$, and staying in $E$ results in the winning probability of $1$ (under some Even's strategy). If an infinite path instead exits $E$, the winning probability is at least the value of the smallest exit of Odd, i.e., $\exit_\tri(E,\ell_i)$. Hence, all states in $S_E$, either staying or exiting, have the winning probability of at least $\min(1,\exit_\tri(E,\ell_i)) = \exit_\tri(E,\ell_i)$. Since $\ell_i \leq \mathcal{V}$, $\exit_\tri(E,\ell_i)$ gives a lower bound for $\mathcal{V}(s)$, i.e., $\exit_\tri(E,\ell_i) \leq \mathcal{V}(s)$ for all $s \in S_E$. As $\ell_i(s) \leq \mathcal{V}(s)$ and $\exit_\tri(E,\ell_i) \leq \mathcal{V}(s)$ for all $s \in S_E$, we obtain $\max(\ell_i(s),\exit_\tri(E,\ell_i)) \leq \mathcal{V}(s)$. Therefore, $\ell_i \leq \mathcal{V}$ holds at the end of the iteration. \qed
\end{proof}

By Proposition~\ref{prop:increasing} and Lemma~\ref{lem:below-V}, the sequence $\ell_i$ is non-decreasing and bounded from above, so the limit $\lim_{i \to \infty} \ell_i$ exists. The dual argument for $u_i$ is straightforward.

\begin{mydefinition}\label{def:limit}
In the setting of Algorithm~\ref{alg:ec-spg-bvi}, we define $\ell^* := \lim_{i \to \infty} \ell_i$ and $u^* := \lim_{i \to \infty} u_i$.
\end{mydefinition}

\begin{proposition}\label{prop:bvi-continuous}
Assume the setting of Algorithm~\ref{alg:ec-spg-bvi}, and let $\bvi : [0,1]^S \times [0,1]^S \to [0,1]^S \times [0,1]^S$ denote the update performed in each iteration of the algorithm (i.e., $\bvi(\ell_i,u_i) = (\ell_{i+1},u_{i+1})$). Then, $\bvi(\ell^*,u^*) = \lim_{i \to \infty} \bvi(\ell_i,u_i)$.
\end{proposition}
\begin{proof}
We work in the complete lattice $([0,1]^S \times [0,1]^S,\leq)$ where $(f,g) \leq (f',g')$ if and only if $f \leq f'$ and $g \geq g'$. We denote the difference between two pairs of functions by $(f,g) - (f',g') := \max_{s \in S} \max((f-f')(s), (g'-g)(s))$.

We first argue $\bvi(\ell^*,u^*) \geq \lim_{i \to \infty} \bvi(\ell_i,u_i)$. By Proposition~\ref{prop:increasing}, $\bvi$ is monotone. For any $i \in \mathbb{N}$, because $(\ell^*,u^*) \geq (\ell_i,u_i)$, we obtain $\bvi(\ell^*,u^*) \geq \bvi(\ell_i,u_i)$. This implies $\bvi(\ell^*,u^*) \geq \lim_{i \to \infty} \bvi(\ell_i,u_i)$.

Next, for any given $\epsilon > 0$, there exists a sufficiently large $n \in \mathbb{N}$ such that $(\ell^*,u^*) - (\ell_n,u_n) < \epsilon$. Moreover, when $\epsilon$ is small enough, one can show that the restricted actions $\av'$ (cf.~Line~\ref{alg2:line5} of Algorithm~\ref{alg:ec-spg-bvi}) constructed under $(\ell^*,u^*)$ and $(\ell_n,u_n)$ are the same. This results in the same subgame $\game'$ being analyzed. By Definition~\ref{def:exit}, the exit value of each BMEC differs by at most $\epsilon$ between the two cases. Hence, we get $\bvi(\ell^*,u^*) - \bvi(\ell_n,u_n) < \epsilon$.

Now we prove $\bvi(\ell^*,u^*) = \lim_{i \to \infty} \bvi(\ell_i,u_i)$. Suppose for the sake of contradiction that this does not hold. Then, $\bvi(\ell^*,u^*) > \lim_{i \to \infty} \bvi(\ell_i,u_i)$. Let $\epsilon := \bvi(\ell^*,u^*) - \lim_{i \to \infty} \bvi(\ell_i,u_i) > 0$. By the aforementioned argument with $\frac{\epsilon}{2} > 0$, there is $n \in \mathbb{N}$ giving $\bvi(\ell^*,u^*) - \bvi(\ell_n,u_n) < \frac{\epsilon}{2}$. This contradicts the fact that this gap equals $\epsilon$, since we have just shown it is strictly less than $\frac{\epsilon}{2}$. \qed
\end{proof}

\begin{proposition}\label{prop:bvi-fixed-point}
Assume the setting of Proposition~\ref{prop:bvi-continuous}. Then, $\bvi(\ell^*,u^*) = (\ell^*,u^*)$, i.e., $(\ell^*,u^*)$ is a fixed point of $\bvi$.
\end{proposition}
\begin{proof}
Following Proposition~\ref{prop:bvi-continuous}, we have
\begin{align*}
\bvi(\ell^*,u^*) = \lim_{i \to \infty}\bvi(\ell_i,u_i) = \lim_{i \to \infty}(\ell_{i+1},u_{i+1}) = \lim_{i \to \infty}(\ell_i,u_i) = (\ell^*,u^*).\tag*{\qed}
\end{align*}
\end{proof}

We now prove that $\ell^*$ and $u^*$ indeed coincide, which implies $\ell^* = \mathcal{V} = u^*$. The proof follows the principle captured by \cite{DBLP:conf/atva/PhalakarnTH25} called the \emph{maximality inheritance principle}.

\begin{lemma}\label{lem:max-inherit}
In the setting of Algorithm~\ref{alg:ec-spg-bvi}, $\ell^* = u^*$.
\end{lemma}
\begin{proof}
Suppose $\ell^* \neq u^*$. Since Lemma~\ref{lem:below-V} and its dual give $\ell^* \leq u^*$, then there is $s \in S$ with $\ell^*(s) < u^*(s)$. Let $M := \{ s^\dagger \in S : \forall s \in S, (u^*-\ell^*)(s^\dagger) \geq (u^*-\ell^*)(s) \} \subseteq S$ be the \emph{gap maximizer set}. Note that $M$ is well-defined due to the finiteness of $S$. For each $s^\dagger \in M$, we now distinguish cases on whether $s^\dagger \in S_\sq$ or $s^\dagger \in S_\tri$, and show that, in both cases, gap maximality is inherited under all actions $a^\dagger \in \av'(s^\dagger)$ defined in Line~\ref{alg2:line5}.

Assume $s^\dagger \in S_\sq$. Since $u^* = \bellman u^*$ (by Proposition~\ref{prop:bvi-fixed-point} and $\bvi$ includes an application of $\bellman$), each action $a^\dagger \in \av'(s^\dagger)$ gives $u^*(s^\dagger) = \phi_{u^*}(s^\dagger,a^\dagger)$. However, this action $a^\dagger$ can be suboptimal for $\ell^*$, giving $\ell^*(s^\dagger) = \max_{a \in \av(s^\dagger)} \phi_{\ell^*}(s^\dagger,a) \geq \phi_{\ell^*}(s^\dagger,a^\dagger)$. So, $(u^*-\ell^*)(s^\dagger) \leq \phi_{u^*}(s^\dagger,a^\dagger)-\phi_{\ell^*}(s^\dagger,a^\dagger) = \phi_{(u^*-\ell^*)}(s^\dagger,a^\dagger) \leq^{(\star)} \max_{s' \in \post(s^\dagger,a^\dagger)}(u^*-\ell^*)(s') \leq (u^*-\ell^*)(s^\dagger)$, where the inequality $\leq^{(\star)}$ is due to the average is no greater than the maximum. Hence, by squeezing, $(u^*-\ell^*)(s^\dagger) = (u^*-\ell^*)(s')$ for all $s' \in \post(s^\dagger,a^\dagger)$. Therefore, we have $\post(s^\dagger,a^\dagger) \subseteq M$ for each $a^\dagger \in \av'(s^\dagger)$.

Now assume $s^\dagger \in S_\tri$. Since $\ell^* = \bellman \ell^*$ (by Proposition~\ref{prop:bvi-fixed-point} and $\bvi$ includes an application of $\bellman$), each action $a^\dagger \in \av'(s^\dagger)$ gives $\ell^*(s^\dagger) = \phi_{\ell^*}(s^\dagger,a^\dagger)$. However, this action $a^\dagger$ can be suboptimal for $u^*$, giving $u^*(s^\dagger) = \min_{a \in \av(s^\dagger)}\phi_{u^*}(s^\dagger,a) \leq \phi_{u^*}(s^\dagger,a^\dagger)$. So, $(u^*-\ell^*)(s^\dagger) \leq \phi_{u^*}(s^\dagger,a^\dagger) - \phi_{\ell^*}(s^\dagger,a^\dagger) \leq (u^*-\ell^*)(s^\dagger)$. Thus, by squeezing, $(u^*-\ell^*)(s^\dagger) = (u^*-\ell^*)(s')$ for all $s' \in \post(s^\dagger,a^\dagger)$. Therefore, we have $\post(s^\dagger,a^\dagger) \subseteq M$ for each $a^\dagger \in \av'(s^\dagger)$.

Combining both cases, we obtain $\forall s^\dagger \in M, \forall a^\dagger \in \av'(s^\dagger), \post(s^\dagger,a^\dagger) \subseteq M$. Since we restrict the available actions to $\av'$ for the subgame $\game'$, any infinite path in $\game'$ starting from $s^\dagger \in M$ remains in $M$. Hence, $M$ contains a state that is Even-winning in $\game'$ (i.e., $M \cap W'_\sq \neq \emptyset$) or $M$ contains a state that is Odd-winning in $\game'$ (i.e., $M \cap W'_\tri \neq \emptyset$) (or both). Without loss of generality, we assume the former case; a symmetric reasoning applies to the latter case.

Suppose $M \cap W'_\sq \neq \emptyset$. Because there is an infinite path, starting from $M \cap W'_\sq$ and staying in $M \cap W'_\sq$ under $\av'$, that makes Even wins, such path must end in some EC, i.e., there is an EC $E = (S_E,\av_E) \in \textsc{BMEC}(W'_\sq,\av')$ with $S_E \subseteq M$. By Lemma~\ref{lem:bmec-win}, $E$ is Even-dominating in both $\game'$ and $\game$.

We then consider the relationship between $\ell^*$ and $\exit_\tri(E,\ell^*)$. If Odd \emph{cannot} exit $E$ (i.e., $\av(s) \setminus \av'(s) = \emptyset$ for all $s \in S_E \cap S_\tri$), then $\exit_\tri(E,\ell^*) = 1$ by Definition~\ref{def:exit}. This implies $\ell^*(s) = 1$ for all $s \in S_E$, contradicting the assumption that $\ell^*(s) < u^*(s) \leq 1$ for all $s \in S_E \subseteq M$.

Otherwise, if Odd \emph{can} exit $E$, then the set $X := \{ s \in S_E \cap S_\tri : \av(s) \setminus \av'(s) \neq \emptyset\}$ is non-empty. Let $s^\flat := \arg\min_{s \in X} \ell^*(s)$ be a state in $X$ with smallest $\ell^*$. By the definition of $\av'$ and $\ell^* = \bellman \ell^*$, for any $s \in X \subseteq S_\tri$, $a' \in \av'(s)$, and $a \in \av(s) \setminus \av'(s)$, we have $\phi_{\ell^*}(s,a) > \phi_{\ell^*}(s,a') = \min_{a' \in \av(s)} \phi_{\ell^*}(s,a') = (\bellman\ell^*)(s) = \ell^*(s) \geq \ell^*(s^\flat)$. Thus, $\exit_\tri(E,\ell^*) > \ell^*(s^\flat)$, and Line~\ref{alg2:line9} modifies $\ell^*(s^\flat)$. This contradicts Proposition~\ref{prop:bvi-fixed-point} that $\ell^*$ is a fixed point of $\bvi$. \qed
\end{proof}

\begin{theorem}
Our proposed $\textnormal{\textsc{SPG-BVI}}(\game,\init,\varepsilon)$ surely terminates and returns a value that is $\varepsilon$-close to $\mathcal{V}(\init)$ for all inputs.
\end{theorem}
\begin{proof}
Follows immediately from Lemma~\ref{lem:below-V} (and its dual) and Lemma~\ref{lem:max-inherit}. \qed
\end{proof}

\subsection{Comparison with \cite{DBLP:conf/lics/KretinskyMW23}}\label{subsec:lics23}

The structure of our proposed Algorithm~\ref{alg:ec-spg-bvi} may appear similar to that of \cite[Algorithm 2]{DBLP:conf/lics/KretinskyMW23}; however, the underlying ideas differ fundamentally. We highlight the key distinctions below.

Firstly, \cite{DBLP:conf/lics/KretinskyMW23} relies on the existence of a \emph{strategy recommender} \cite[Definition 1]{DBLP:conf/lics/KretinskyMW23}. A strategy recommender is a procedure that generates a sequence $(\sigma_{\sq,i},\sigma_{\tri,i})$ such that, for some $i \in \mathbb{N}$, the pair is optimal. Such recommenders can be constructed for a range of objectives, including reachability, safety, total reward, and mean-payoff, since these objectives can be characterized as least or greatest fixed points of suitable operators. However, this approach does not extend to parity objectives. In contrast, our algorithm and correctness proofs do not rely on a strategy recommender and depend solely on restricted actions $\av'$.

Secondly, \cite{DBLP:conf/lics/KretinskyMW23} adjusts the two bound sequences using reasoning based on MDPs. For example, in constructing the bottom-up sequence $\ell_i$, they fix a strategy for Maximizer, thereby inducing an MDP controlled only by Minimizer, and compute an optimal response for Minimizer. Since the fixed Maximizer's strategy may be suboptimal, the resulting winning probabilities can be smaller, yielding a sound lower bound. By leveraging a strategy recommender, these fixed strategies eventually become optimal, ensuring convergence of the bounds.

In contrast, rather than relying on MDP-based reasoning, we solve the qualitative SPG problem (Problem~\ref{prob:qual-prob}) to determine the winning sets of both players within the restricted subgame $\game'$. By systematically refining the bounds of these winning sets, we establish convergence of our algorithm without requiring a strategy recommender. From this perspective, our approach can be viewed as a generalization of \cite{DBLP:conf/lics/KretinskyMW23}.

\section{Conclusions and Future Work}\label{sec:conclusion}

We introduced the first BVI algorithm for SPGs, enabling the computation of winning probabilities under parity objectives with precision guarantees. A central difficulty is that, unlike reachability or safety objectives, its value function is not characterized as either the least or greatest fixed point of the Bellman operator. Our technique resolves this by combining qualitative analysis with quantitative updates. The resulting algorithm is sound and converges with guarantees.

As future work, we aim to extend the framework to richer objectives, particularly quantitative rewards. We also plan to develop a deeper lattice-theoretic foundation for the underlying procedure, with the goal of characterizing fixed point uniqueness and enabling systematic generalizations of the approach.

\begin{credits}
\subsubsection{\ackname} 
This work is supported by the ASPIRE grant No. JPMJAP2301, JST.
\end{credits}

\pagebreak
\bibliographystyle{splncs04}
\bibliography{ref}

@inproceedings{DBLP:conf/fossacs/Zielonka04,
  author       = {Wieslaw Zielonka},
  title        = {Perfect-Information Stochastic Parity Games},
  booktitle    = {{FoSSaCS 2004}},
  series       = {LNCS},
  volume       = {2987},
  pages        = {499--513},
  publisher    = {Springer},
  year         = {2004}
}

@article{DBLP:journals/tac/CourcoubetisY98,
  author       = {Costas Courcoubetis and
                  Mihalis Yannakakis},
  title        = {Markov decision processes and regular events},
  journal      = {{IEEE} Trans. Autom. Control.},
  volume       = {43},
  number       = {10},
  pages        = {1399--1418},
  year         = {1998}
}

@phdthesis{DBLP:phd/us/Alfaro97,
  author       = {Luca de Alfaro},
  title        = {Formal verification of probabilistic systems},
  school       = {Stanford University, {USA}},
  year         = {1997}
}

@inproceedings{DBLP:conf/soda/ChatterjeeJH04,
  author       = {Krishnendu Chatterjee and
                  Marcin Jurdzinski and
                  Thomas A. Henzinger},
  title        = {Quantitative stochastic parity games},
  booktitle    = {{SODA 2004}},
  pages        = {121--130},
  publisher    = {{SIAM}},
  year         = {2004}
}

@inproceedings{DBLP:conf/csl/ChatterjeeJH03,
  author       = {Krishnendu Chatterjee and
                  Marcin Jurdzinski and
                  Thomas A. Henzinger},
  title        = {Simple Stochastic Parity Games},
  booktitle    = {{CSL 2003}},
  series       = {LNCS},
  volume       = {2803},
  pages        = {100--113},
  publisher    = {Springer},
  year         = {2003}
}

@inproceedings{DBLP:conf/isaac/AnderssonM09,
  author       = {Daniel Andersson and
                  Peter Bro Miltersen},
  title        = {The Complexity of Solving Stochastic Games on Graphs},
  booktitle    = {{ISAAC 2009}},
  series       = {LNCS},
  volume       = {5878},
  pages        = {112--121},
  publisher    = {Springer},
  year         = {2009}
}

@inproceedings{DBLP:conf/concur/BerthonKZ25,
  author       = {Rapha{\"{e}}l Berthon and
                  Joost{-}Pieter Katoen and
                  Zihan Zhou},
  title        = {A Direct Reduction from Stochastic Parity Games to Simple Stochastic Games},
  booktitle    = {{CONCUR 2025}},
  series       = {LIPIcs},
  volume       = {348},
  pages        = {9:1--9:21},
  publisher    = {Schloss Dagstuhl - Leibniz-Zentrum f{\"{u}}r Informatik},
  year         = {2025}
}

@inproceedings{DBLP:conf/stacs/ChatterjeeH06,
  author       = {Krishnendu Chatterjee and
                  Thomas A. Henzinger},
  title        = {Strategy Improvement and Randomized Subexponential Algorithms for Stochastic Parity Games},
  booktitle    = {{STACS 2006}},
  series       = {LNCS},
  volume       = {3884},
  pages        = {512--523},
  publisher    = {Springer},
  year         = {2006}
}

@inproceedings{DBLP:conf/cav/HahnSTZ16,
  author       = {Ernst Moritz Hahn and
                  Sven Schewe and
                  Andrea Turrini and
                  Lijun Zhang},
  title        = {A Simple Algorithm for Solving Qualitative Probabilistic Parity Games},
  booktitle    = {{CAV 2016}},
  series       = {LNCS},
  volume       = {9780},
  pages        = {291--311},
  publisher    = {Springer},
  year         = {2016}
}

@inproceedings{DBLP:conf/vmcai/HahnST017,
  author       = {Ernst Moritz Hahn and
                  Sven Schewe and
                  Andrea Turrini and
                  Lijun Zhang},
  title        = {Synthesising Strategy Improvement and Recursive Algorithms for Solving 2.5 Player Parity Games},
  booktitle    = {{VMCAI 2017}},
  series       = {LNCS},
  volume       = {10145},
  pages        = {266--287},
  publisher    = {Springer},
  year         = {2017}
}

@article{DBLP:journals/tcs/Zielonka98,
  author       = {Wieslaw Zielonka},
  title        = {Infinite Games on Finitely Coloured Graphs with Applications to Automata on Infinite Trees},
  journal      = {Theor. Comput. Sci.},
  volume       = {200},
  number       = {1-2},
  pages        = {135--183},
  year         = {1998}
}

@article{DBLP:journals/apal/McNaughton93,
  author       = {Robert McNaughton},
  title        = {Infinite Games Played on Finite Graphs},
  journal      = {Ann. Pure Appl. Logic},
  volume       = {65},
  number       = {2},
  pages        = {149--184},
  year         = {1993}
}

@article{DBLP:journals/siamcomp/CaludeJKLS22,
  author       = {Cristian S. Calude and
                  Sanjay Jain and
                  Bakhadyr Khoussainov and
                  Wei Li and
                  Frank Stephan},
  title        = {Deciding Parity Games in Quasi-polynomial Time},
  journal      = {{SIAM} J. Comput.},
  volume       = {51},
  number       = {2},
  pages        = {17--152},
  year         = {2022}
}

@inproceedings{DBLP:conf/lics/JurdzinskiL17,
  author       = {Marcin Jurdzinski and
                  Ranko Lazic},
  title        = {Succinct progress measures for solving parity games},
  booktitle    = {{LICS 2017}},
  pages        = {1--9},
  publisher    = {{IEEE} Computer Society},
  year         = {2017}
}

@inproceedings{DBLP:conf/mfcs/Parys19,
  author       = {Pawel Parys},
  title        = {Parity Games: Zielonka's Algorithm in Quasi-Polynomial Time},
  booktitle    = {{MFCS 2019}},
  series       = {LIPIcs},
  volume       = {138},
  pages        = {10:1--10:13},
  publisher    = {Schloss Dagstuhl - Leibniz-Zentrum f{\"{u}}r Informatik},
  year         = {2019}
}

@inproceedings{DBLP:conf/lics/Lehtinen18,
  author       = {Karoliina Lehtinen},
  title        = {A modal {\(\mu\)} perspective on solving parity games in quasi-polynomial time},
  booktitle    = {{LICS 2018}},
  pages        = {639--648},
  publisher    = {{ACM}},
  year         = {2018}
}

@article{DBLP:journals/jcss/AlfaroM04,
  author       = {Luca de Alfaro and
                  Rupak Majumdar},
  title        = {Quantitative solution of omega-regular games},
  journal      = {J. Comput. Syst. Sci.},
  volume       = {68},
  number       = {2},
  pages        = {374--397},
  year         = {2004}
}

@article{DBLP:journals/jsyml/Martin98,
  author       = {Donald A. Martin},
  title        = {The Determinacy of {Blackwell} Games},
  journal      = {J. Symb. Log.},
  volume       = {63},
  number       = {4},
  pages        = {1565--1581},
  year         = {1998}
}

@article{DBLP:journals/siamcomp/JurdzinskiPZ08,
  author       = {Marcin Jurdzinski and
                  Mike Paterson and
                  Uri Zwick},
  title        = {A Deterministic Subexponential Algorithm for Solving Parity Games},
  journal      = {{SIAM} J. Comput.},
  volume       = {38},
  number       = {4},
  pages        = {1519--1532},
  year         = {2008}
}

@inproceedings{DBLP:conf/atva/FriedmannL09,
  author       = {Oliver Friedmann and
                  Martin Lange},
  title        = {Solving Parity Games in Practice},
  booktitle    = {{ATVA 2009}},
  series       = {LNCS},
  volume       = {5799},
  pages        = {182--196},
  publisher    = {Springer},
  year         = {2009}
}

@inproceedings{DBLP:conf/cav/KelmendiKKW18,
  author       = {Edon Kelmendi and
                  Julia Kr{\"{a}}mer and
                  Jan Kret{\'{\i}}nsk{\'{y}} and
                  Maximilian Weininger},
  title        = {Value Iteration for Simple Stochastic Games{:} Stopping Criterion and Learning Algorithm},
  booktitle    = {{CAV 2018}},
  series       = {LNCS},
  volume       = {10981},
  pages        = {623--642},
  publisher    = {Springer},
  year         = {2018}
}

@inproceedings{DBLP:conf/lics/KretinskyMW23,
  author       = {Jan Kret{\'{\i}}nsk{\'{y}} and
                  Tobias Meggendorfer and
                  Maximilian Weininger},
  title        = {Stopping Criteria for Value Iteration on Stochastic Games with Quantitative Objectives},
  booktitle    = {{LICS 2023}},
  pages        = {1--14},
  publisher    = {{IEEE}},
  year         = {2023}
}

@inproceedings{DBLP:conf/cav/MeggendorferW24,
  author       = {Tobias Meggendorfer and
                  Maximilian Weininger},
  title        = {Playing Games with Your {PET:} Extending the Partial Exploration Tool to Stochastic Games},
  booktitle    = {{CAV 2024}},
  series       = {LNCS},
  volume       = {14683},
  pages        = {359--372},
  publisher    = {Springer},
  year         = {2024}
}

@inproceedings{DBLP:conf/atva/AzeemEKSW22,
  author       = {Muqsit Azeem and
                  Alexandros Evangelidis and
                  Jan Kret{\'{\i}}nsk{\'{y}} and
                  Alexander Slivinskiy and
                  Maximilian Weininger},
  title        = {Optimistic and Topological Value Iteration for Simple Stochastic Games},
  booktitle    = {{ATVA 2022}},
  series       = {LNCS},
  volume       = {13505},
  pages        = {285--302},
  publisher    = {Springer},
  year         = {2022}
}

@inproceedings{DBLP:conf/gandalf/AzeemKW25,
  author       = {Muqsit Azeem and
                  Jan Kret{\'{\i}}nsk{\'{y}} and
                  Maximilian Weininger},
  title        = {Sound Value Iteration for Simple Stochastic Games},
  booktitle    = {{GandALF 2025}},
  series       = {EPTCS},
  volume       = {428},
  pages        = {29--44},
  publisher    = {Open Publishing Association},
  year         = {2025}
}

@inproceedings{DBLP:conf/cav/PhalakarnTHH20,
  author       = {Kittiphon Phalakarn and
                  Toru Takisaka and
                  Thomas Haas and
                  Ichiro Hasuo},
  title        = {Widest Paths and Global Propagation in Bounded Value Iteration for Stochastic Games},
  booktitle    = {{CAV 2020}},
  series       = {LNCS},
  volume       = {12225},
  pages        = {349--371},
  publisher    = {Springer},
  year         = {2020}
}

@inproceedings{DBLP:conf/atva/PhalakarnTH25,
  author       = {Kittiphon Phalakarn and
                  Yun Chen Tsai and
                  Ichiro Hasuo},
  title        = {Widest Path Games and Maximality Inheritance in Bounded Value Iteration for Stochastic Games},
  booktitle    = {{ATVA 2025}},
  series       = {LNCS},
  volume       = {16145},
  pages        = {109--131},
  publisher    = {Springer},
  year         = {2025}
}

@inproceedings{DBLP:conf/vmcai/FuHLSSTZ22,
  author       = {Chen Fu and
                  Ernst Moritz Hahn and
                  Yong Li and
                  Sven Schewe and
                  Meng Sun and
                  Andrea Turrini and
                  Lijun Zhang},
  title        = {{EPMC} Gets Knowledge in Multi-agent Systems},
  booktitle    = {{VMCAI 2022}},
  series       = {LNCS},
  volume       = {13182},
  pages        = {93--107},
  publisher    = {Springer},
  year         = {2022}
}

@book{DBLP:books/daglib/0020348,
  author       = {Christel Baier and
                  Joost{-}Pieter Katoen},
  title        = {Principles of model checking},
  publisher    = {{MIT} Press},
  year         = {2008}
}

@article{DBLP:journals/dm/Baranga91,
  author       = {Andrei Baranga},
  title        = {The contraction principle as a particular case of {Kleene}'s fixed point theorem},
  journal      = {Discret. Math.},
  volume       = {98},
  number       = {1},
  pages        = {75--79},
  year         = {1991}
}

@article{DBLP:journals/tcs/HaddadM18,
  author       = {Serge Haddad and
                  Benjamin Monmege},
  title        = {Interval iteration algorithm for {MDPs} and {IMDPs}},
  journal      = {Theor. Comput. Sci.},
  volume       = {735},
  pages        = {111--131},
  year         = {2018}
}

@article{DBLP:journals/ipl/ChatterjeeH08,
  author       = {Krishnendu Chatterjee and
                  Thomas A. Henzinger},
  title        = {Reduction of stochastic parity to stochastic mean-payoff games},
  journal      = {Inf. Process. Lett.},
  volume       = {106},
  number       = {1},
  pages        = {1--7},
  year         = {2008}
}

@article{DBLP:journals/iandc/Condon92,
  author       = {Anne Condon},
  title        = {The Complexity of Stochastic Games},
  journal      = {Inf. Comput.},
  volume       = {96},
  number       = {2},
  pages        = {203--224},
  year         = {1992}
}

@inproceedings{DBLP:conf/spin/ChatterjeeH08,
  author       = {Krishnendu Chatterjee and
                  Thomas A. Henzinger},
  title        = {Value Iteration},
  booktitle    = {25 Years of Model Checking},
  series       = {LNCS},
  volume       = {5000},
  pages        = {107--138},
  publisher    = {Springer},
  year         = {2008}
}

@article{DBLP:journals/ejcon/SvorenovaK16,
  author       = {Mar{\'{\i}}a Svorenov{\'{a}} and
                  Marta Kwiatkowska},
  title        = {Quantitative verification and strategy synthesis for stochastic games},
  journal      = {Eur. J. Control},
  volume       = {30},
  pages        = {15--30},
  year         = {2016}
}

@article{DBLP:journals/access/TusharYSNASP23,
  author       = {Wayes Tushar and
                  Chau Yuen and
                  Tapan Kumar Saha and
                  Sohrab Nizami and
                  Mollah R. Alam and
                  David B. Smith and
                  H. Vincent Poor},
  title        = {A Survey of Cyber-Physical Systems From a Game-Theoretic Perspective},
  journal      = {{IEEE} Access},
  volume       = {11},
  pages        = {9799--9834},
  year         = {2023}
}

@article{DBLP:journals/arcras/MardenS18,
  author       = {Jason R. Marden and
                  Jeff S. Shamma},
  title        = {Game Theory and Control},
  journal      = {Annu. Rev. Control. Robotics Auton. Syst.},
  volume       = {1},
  pages        = {105--134},
  year         = {2018}
}

@inproceedings{DBLP:conf/tacas/HartmannsJQW23,
  author       = {Arnd Hartmanns and
                  Sebastian Junges and
                  Tim Quatmann and
                  Maximilian Weininger},
  title        = {A Practitioner's Guide to {MDP} Model Checking Algorithms},
  booktitle    = {{TACAS 2023}},
  series       = {LNCS},
  volume       = {13993},
  pages        = {469--488},
  publisher    = {Springer},
  year         = {2023}
}

@article{revised,
  author       = {Arnd Hartmanns and
                  Sebastian Junges and
                  Tim Quatmann and
                  Maximilian Weininger},
  title        = {The Revised Practitioner's Guide to {MDP} Model Checking Algorithms},
  journal      = {Int. J. Softw. Tools Technol. Transfer},
  year         = {2026}
}

@article{DBLP:journals/iandc/EisentrautKKW22,
  author       = {Julia Eisentraut and
                  Edon Kelmendi and
                  Jan Kret{\'{\i}}nsk{\'{y}} and
                  Maximilian Weininger},
  title        = {Value iteration for simple stochastic games: Stopping criterion and
                  learning algorithm},
  journal      = {Inf. Comput.},
  volume       = {285},
  number       = {Part},
  pages        = {104886},
  year         = {2022}
}

\end{document}